\documentclass[a4paper,USenglish,cleveref, autoref, thm-restate]{compgeomtop}

\usepackage[utf8]{inputenc}
\usepackage[T1]{fontenc}

\usepackage[linesnumbered,ruled,onelanguage,vlined]{algorithm2e}

\usepackage{tikz}

\usetikzlibrary{calc}
\let\le\leqslant
\let\ge\geqslant

\title{A note on the flip distance between non-crossing spanning trees}

\author{Nicolas Bousquet}{Univ. Lyon, Université Lyon 1, CNRS, LIRIS UMR CNRS 5205, F-69621, Lyon, France}{nicolas.bousquet@cnrs.fr}{https://orcid.org/0000-0003-0170-0503}{} 

\author{Valentin Gledel}{Department of Mathematics and Mathematical Statistics, Ume\r{a} University, Sweden}{valentin.gledel@umu.se}{https://orcid.org/0000-0003-4736-4656}{}

\author{Jonathan Narboni}{Theoretical Computer Science Department, Faculty of Mathematics and Computer Science,
Jagiellonian University, Kraków, Poland}{jonathan.narboni@uj.edu.pl}{https://orcid.org/0000-0002-3087-5073}{National Science Center of Poland grant 2019/34/E/ST6/00443}
\author{Théo Pierron}{Univ. Lyon, Université Lyon 1, CNRS, LIRIS UMR CNRS 5205, F-69621, Lyon, France}{theo.pierron@univ-lyon1.fr}{https://orcid.org/0000-0002-5586-5613}{}

\authorrunning{N. Bousquet, V. Gledel, J. Narboni and T. Pierron} 

\Copyright{Nicolas Bousquet, Valentin Gledel, Jonathan Narboni and Théo Pierron} 


\relatedversion{}

\funding{ANR project GrR (ANR-18-CE40-0032)}

\Volume{42}
\ArticleNo{23}

\begin{document}

\maketitle

\begin{abstract}
We consider spanning trees of $n$ points in convex position whose edges are pairwise non-crossing. Applying a flip to such a tree consists in adding an edge and removing another so that the result is still a non-crossing spanning tree. Given two trees, we investigate the minimum number of flips required to transform one into the other. The naive $2n-\Omega(1)$ upper bound stood for 25 years until a recent breakthrough from Aichholzer et al. yielding a $2n-\Omega(\log n)$ bound. We improve their result with a $2n-\Omega(\sqrt{n})$ upper bound, and we strengthen and shorten the proofs of several of their results.
\end{abstract}

\section{Introduction}

We fix a set  $P=\{v_1,\ldots,v_n\}$ of $n$ points in the plane in convex position (and we assume that $v_1,\ldots,v_n$ appear in this order on the convex hull of $P$). A non-crossing spanning tree is a spanning tree of $P$ whose edges are straight line segments between pairs of points such that no two
edges intersect (except on their endpoints). A \emph{flip} removes an edge of a non-crossing spanning tree and adds another one so that the result is again a non-crossing spanning tree of $P$. 
A \emph{flip sequence} is a sequence of non-crossing spanning trees such that consecutive spanning trees in the sequence differ by exactly one flip. 
We study the problem of transforming a non-crossing spanning tree into another via a sequence of flips. 

Given two non-crossing spanning trees $T_1$ and $T_2$, observe that the size $|T_1\Delta T_2|$ of the symmetric difference between their sets of edges may decrease by at most $2$ when applying a flip, hence $|T_1\Delta T_2|/2$ flips are required (note that this quantity can be as large as $n$ when $T_1$ and $T_2$ have no common edge). Hernando et al.~\cite{HernandoHMMN99} proved that there exists two trees $T_1$ and $T_2$ such that any flip sequence needs at least $\frac 32 n-5$ flips. Regarding upper bounds, Avis and Fukuda~\cite{AvisF96} proved that there always exists a flip sequence between $T_1$ and $T_2$ using at most $2n-4$ flips. This simple $2n-\Omega(1)$ upper bound was not improved in the last $25$ years until a recent work of Aichholzer et al.~\cite{Aichholzer22+} who improved the upper bound into $2n-\Omega(\log n)$. They also proved that there exists a flip sequence of length $\frac 32n$ if one of the two spanning trees is a path. 

In this paper, we improve the upper bound of~\cite{Aichholzer22+} by proving that there exists a transformation of length at most $2n-\Omega(\sqrt{n})$ (Corollary~\ref{coro1}). We also reprove with a shorter proof the existence of a transformation of length $\frac 32n$ when one tree is a path (Theorem~\ref{thm:onepath}) and relax this result by showing that if one of the trees contains an induced path of length $t$ then there exists a transformation of length $2n-\frac t3$ (Corollary~\ref{coro:longpath}).

Finally, we prove that if one of the trees is a nice caterpillar (whose precise definition will be given in Section~\ref{sec:result}), the shortest transformation has length at most $\frac 32 n$ (Corollary~\ref{coro:worst}). This is remarkable since, as far as we know, in all the examples where $\frac 32 n - \Omega(1)$ flips are needed, at least one of the two non-crossing spanning trees is a nice caterpillar~\cite{HernandoHMMN99,Aichholzer22+}. So our statement essentially ensures that, if $\frac 32n$ is not the tight upper bound, then the spanning trees between which a larger transformation is needed should be constructed quite differently. 

We think that all our partial results give additional credit to the following conjecture:

\begin{conjecture}
 There is a flip sequence between any pair of non-crossing spanning trees of length at most $\frac 32 n$. 
\end{conjecture}

All our proofs are simple, self-contained, and mainly follow from simple applications of a lemma stated at the beginning of the next section.

\section{Results}\label{sec:result}

Recall that all along the paper we consider a set of $n$ points $v_1,\ldots,v_n$ in convex position appearing in that order. 
A \emph{leaf} of a tree $T$ is a vertex of degree one. An \emph{internal node} of $T$ is a vertex that is not a leaf. A \emph{border edge} is an edge of the convex hull, \emph{i.e.} $v_iv_{i+1}$ for some $i$.
Let us first prove the following claim:

\begin{claim}\label{clm:simple}
Let $T$ be a non-crossing spanning tree and $e$ be a border edge. Then we can add $e$ in $T$ with an edge-flip without removing any border edge (except if $T$ only contains border edges).
\end{claim}
\begin{proof}
Adding $e$ to $T$ does not create any crossing, since $e$ belongs to the convex hull of $P$. Moreover, the unique cycle in $T\cup\{e\}$ must contain at least an edge $e'$ that does not belong to set of border edges, since otherwise $T\cup\{e\}$ is precisely the convex hull of $P$. 
\end{proof}

All our results follow from the following simple but very useful lemma:

\begin{restatable}[]{lemma}{main}\label{lem:main}
Let $i \le n$. Let $T_1,T_2$ be two spanning trees of $P$ such that $T_1$ contains all the edges $v_jv_{j+1}$ for $j <i$ and $T_2$ has no edge $v_jv_k$ with $j> i$ and $k>i$. Then there exists a flip sequence between $T_1$ and $T_2$ of length at most $|T_1\Delta T_2|/2$.
\end{restatable}
 
\begin{proof}
Let us denote by $X$ the subset of points $\{v_1,\ldots,v_i\}$ and $E_X$ the set of edges $v_jv_{j+1}$ for every $j <i$. We first apply Claim~\ref{clm:simple} to $T_2$ until it contains all the edges of $E_X$. Since $T_1$ contains all edges in $E_X$, we can choose these flips in such a way the symmetric difference $T_1\Delta T_2$ decreases by $2$ at each time. Note that, afterwards, all the vertices outside of $X$ are leaves in $T_2$. 

We now prove that, as long as $T_1\neq T_2$, one can find a \emph{good} flip in $T_1$, \emph{i.e.} such that after applying it to $T_1$, the resulting tree $T'_1$ still satisfies the hypothesis of the lemma, and $|T'_1\Delta T_2|=|T_1\Delta T_2|-2$. The conclusion immediately follows by iterating this argument on $T'_1$ until we reach $T_2$. 

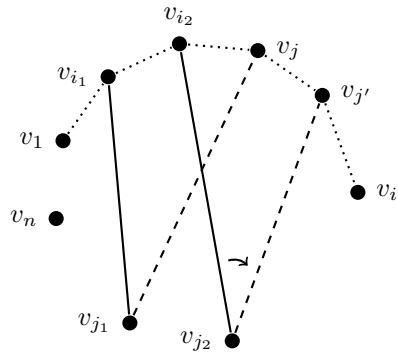
\begin{figure}[!ht]
\centering
\begin{tikzpicture}[thick, v/.style={fill, circle,inner sep=2pt}]
    \node[v,label=left:{$v_1$}] (v1) at (160:2) {};
    \node[v,label=left:{$v_{i_1}$}] (v2) at (130:2) {};
    \node[v,label=above:{$v_{i_2}$}] (v3) at (100:2) {};
    \node[v,label=right:{$v_j$}] (v4) at (70:2) {};
    \node[v,label=right:{$v_{j'}$}] (v5) at (40:2) {};
    \node[v,label=right:{$v_i$}] (v6) at (0:2) {};
    \node[v,label=left:{$v_{j_2}$}] (v7) at (280:2) {};
    \node[v,label=left:{$v_{j_1}$}] (v8) at (240:2) {};
    \node[v,label=left:{$v_n$}] (v9) at (190:2) {};
    \draw[dotted] (v1) -- (v2) -- (v3) -- (v4) -- (v5) -- (v6);
    \draw (v2) -- (v8);
    \draw (v3) to node[midway,near end] (A) {}  (v7);
    \draw[dashed] (v8) -- (v4);
    \draw[dashed] (v7) to node[midway,near start] (B) {} (v5);
    \draw[bend left, ->] (A) to (B);
\end{tikzpicture}
\caption{Dotted edges denote the path $E_X$, edges of $T_2$ are dashed and edges of $T_1$ are full. The flip $v_{i_1}v_{j_1}\to v_{j_1}v_j$ is not good, but $v_{i_2}v_{j_2}\to v_{j_2}v_{j'}$ is.}
\label{fig:lem2}
\end{figure}
Assume that $T_1\neq T_2$, and let $v_{i_1}v_{j_1}$ be the leftmost edge of $T_1\setminus T_2$ not in $E_X$, \emph{i.e.} such that $i_1 \in[1,i]$ is minimum and $j_1>i$ is maximum (see Figure~\ref{fig:lem2}). Such an edge must exist since $T_1\neq T_2$. Recall that $v_{j_1}$ is a leaf in $T_2$ and its parent $v_j$ must lie in $X \setminus\{v_{i_1}\}$. If exchanging $v_{i_1}v_{j_1}$ with $v_{j_1}v_{j}$ in $T_1$ is a good flip, then we are done. So we may assume that there must be an edge from $T_1$ that crosses $v_{j_1}v_j$, say $v_{i_2}v_{j_2}$. By minimality of $i_1$ and since $T_1$ is a non-crossing tree, we thus have $j>i_1$, $j_2<j_1$, and hence $i_2\geqslant i_1$. Iterating the previous argument with $v_{i_2}v_{j_2}$ instead of $v_{i_1}v_{j_1}$ yields either a good flip, or a sequence of vertices $v_{j_1},\ldots,v_{j_k}$ with $j_1>\cdots >j_k$. Note that $j_k > i$ since otherwise $v_{i_k}v_{j_k}\cup E_X$ induces a cycle in $T_2$. Therefore, this process must terminate after at most $n-i$ steps, so there exists a good flip, which concludes.
\end{proof}

In the rest of the paper, we derive corollaries from that lemma. 
First, we improve the upper bound of~\cite{Aichholzer22+}:

\begin{corollary}\label{coro1}
There exists a flip sequence of length at most $2n - \Omega(\sqrt n)$ between any pair of non-crossing spanning trees.
\end{corollary}
\begin{proof}
Let $T_1,T_2$ be two non-crossing spanning trees. Partition arbitrarily the set $P$ into $\sqrt{n}$ sections of size $\sqrt{n}$. If one section does not contain any edge of $T_1$ with both endpoints in it, then we use Claim~\ref{clm:simple} to add to $T_2$ all the border edges outside of this section (in $n-\sqrt{n}$ flips), and then apply Lemma~\ref{lem:main} to transform the resulting tree into $T_1$ with $n$ additional flips.

Therefore, we can assume that all the sections contain an edge with both endpoints in $T_1$ and thus a border edge. Applying again $n-\sqrt{n}$ times Claim~\ref{clm:simple} to $T_1$ and $n$ times to $T_2$, we can transform both trees into a tree only containing border edges. This yields a flip sequence between $T_1$ and $T_2$ in at most $2n-\sqrt{n}+1$ steps (since any two trees only containing border edges are adjacent via a single edge-flip).
\end{proof}

A \emph{caterpillar} is a tree such that the set of internal nodes induces a path. A caterpillar is \emph{nice} if for every four path on four vertices $u,v,w,x$ such that $v,w$ are internal nodes, 
the line $vw$ splits the convex hull of $P$ in two parts, one of them containing $u$ and the other $x$ (see Figure~\ref{fig:caterpillar}). 

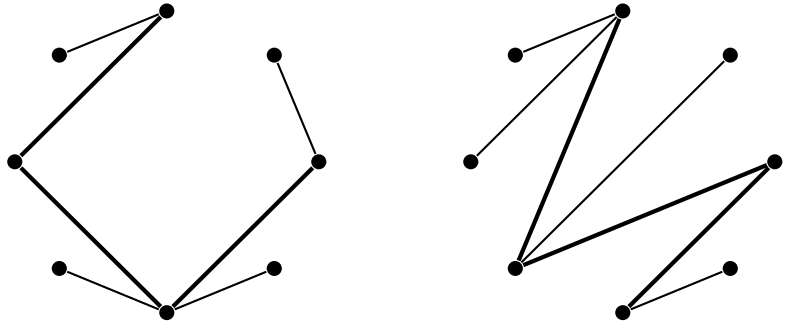
\begin{figure}[!ht]
\centering
\begin{tikzpicture}[thick, every node/.style={fill, circle,inner sep=2pt}]
\node (A) at (90:2) {};
\node (B) at (135:2) {};
\node (C) at (180:2) {};
\node (D) at (225:2) {};
\node (E) at (270:2) {};
\node (F) at (315:2) {};
\node (G) at (0:2) {};
\node (H) at (45:2) {};
\draw (B) -- (A); 
\draw (D) -- (E) -- (F);
\draw (G) -- (H);
\draw[ultra thick] (A) -- (C) -- (E) -- (G);
\tikzset{xshift=6cm}
\node (A) at (90:2) {};
\node (B) at (135:2) {};
\node (C) at (180:2) {};
\node (D) at (225:2) {};
\node (E) at (270:2) {};
\node (F) at (315:2) {};
\node (G) at (0:2) {};
\node (H) at (45:2) {};
\draw (B) -- (A) -- (C); 
\draw (D) -- (H);
\draw (E) -- (F);
\draw[ultra thick] (A) -- (D) -- (G) -- (E);
\end{tikzpicture}
\caption{Two caterpillars, whose internal path is highlighted in bold. The left one is not nice, while the right one is.}
\label{fig:caterpillar}
\end{figure}

As far as we know, in all the examples where $\frac 32 n - \Omega(1)$ flips are needed, at least one of the two non-crossing spanning trees is a nice caterpillar~\cite{HernandoHMMN99,Aichholzer22+}. 

\begin{corollary}\label{coro:worst}
Let $T_1,T_2$ be non-crossing spanning trees such that $T_1$ is a nice caterpillar. There exists a flip sequence between $T_1$ and $T_2$ of length at most $\frac 32 n$.
\end{corollary}

\begin{proof}
Let us denote by $w_1,\ldots,w_k$ the set of internal nodes of $T_1$. By symmetry, assume that $w_1=v_1$ and denote by $j$ the index such that $v_j=w_k$. Up to reversing the ordering of the vertices, we may also assume that $j\leqslant n/2$. Applying Claim~\ref{clm:simple} at most $j$ times, we can add all edges $v_iv_{i+1}$ for $i<j$ to $T_2$. Note that the obtained tree satisfies the hypothesis of Lemma~\ref{lem:main}, hence we can transform it into $T_1$ in at most $n$ steps, which completes the proof.
\end{proof}

Let $P$ be a set of points in convex position.
We say that $t$ edges $a_1b_1,\ldots,a_tb_t$ are \emph{parallel} if $a_1,a_2,\ldots,a_t,b_t,b_{t-1},\ldots,b_1$ appear in that order in the cyclic ordering corresponding to the convex hull of $P$. If moreover all their endpoints are pairwise distinct, then we say the edges are \emph{strictly parallel}. Using Lemma~\ref{lem:main}, one can also easily prove the following that will lead to another interesting corollary:

\begin{lemma}\label{lem:paraedge}
Let $T_1,T_2$ be two non-crossing spanning trees such that $T_1$ has $t$ parallel edges (resp. strictly parallel edges). There exists a flip sequence between $T_1$ and $T_2$ of length at most $2 n-\frac{t-1}{2}$ (resp. $2n-t+1$).
\end{lemma}
\begin{proof}
Let $t'$ be the maximum number of parallel edges in $T_1$. Note that $t' \ge t$.
Let $e_1=a_1b_1,\ldots,e_{t'}=a_{t'}b_{t'}$ be $t'$ parallel edges of $T_1$. We say that vertices between $b_{t'}$ and $b_1$ (resp. $a_1$ and $a_{t'}$) in the cyclic ordering in the section that does not contain $a_1$ (resp. $b_1$) are \emph{top vertices} (resp. \emph{bottom vertices}). Let $b$ be the number of bottom vertices.

Since $T_1$ is non-crossing, for every $i \le t'-1$, there is a path $Q_i$ from an endpoint of $e_i$ to an endpoint of $e_{i+1}$ in $T_1$. Note that this path might be reduced to a single vertex if the two edges share an endpoint (we say that $Q_i$ is \emph{trivial}). Observe that the same trivial path may appear several times if several edges share the same endpoint. By maximality of $t'$, there cannot be an edge in $Q_i$ between a top vertex and a bottom vertex. Therefore we can classify the $t'-1$ paths $Q_i$ in two types: $Q_i$ is a \emph{top path} if it only contains edges between top vertices and a \emph{bottom path} otherwise. By symmetry, we can assume that at least $w:=(t'-1)/2$ paths are top paths. 
    
We claim that we can transform $T_1$ into a tree $T_1'$ such that all the edges of the tree $T'_1$ have at least one endpoint between $a_1$ and $a_{t'}$ in at most $b-w$ steps. 

Recall that by maximality of $t'$ there is no edge between the top part $\{a_i,\ldots,a_{i+1}\}$ and the bottom part $B_i=\{b_{i+1},\ldots,b_i \}$ except $a_ib_i$ and $a_{i+1}b_{i+1}$.
If $Q_i$ is a bottom path, we can remove one edge of the bottom part and add one edge in the top part to get a top path. 
Now, we say that an edge $bb'$ with $b < b'$ with both endpoints in $B_i$ is \emph{exterior} if no edge $cc'$ distinct from $bb'$ with both endpoints in $B_i$ satisfies $c \le b < b' \le c'$.
One can easily remark that we can iteratively replace an exterior edge $bb'$ by an arc connecting $b$ or $b'$ to $a_i$ until no edge with both endpoints in $B_i$ remains. So we can ensure that no edge with both endpoints in $B_i$ remains in at most $|B_i|-1$ steps ($-2$ if $Q_i$ was initially a top path). If we sum over all the sections, since $\sum_i (|B_i|-1)= b$ and we remove $1$ additional flip for each of the $w$ top paths, this process yields a tree $T'_1$ where all the edges have at least one endpoint between $a_1$ and $a_{t'}$ in $b-w$ flips.

Now we can transform $T_2$ into a tree $T_2'$ that contains all border edges except maybe between bottom vertices in at most $n-b$ steps by Claim~\ref{clm:simple}. Finally, we may apply Lemma~\ref{lem:main} to transform $T_1'$ into $T_2'$ in at most $n$ steps, which in total gives a flip sequence of length at most $(b-w)+(n-b)+n = 2n-w$, as claimed.

In the strictly parallel case, observe that each non-trivial top path must contain a border edge in $T_1$ (and then in $T'_1$). Note that there are at most $t'-t$ trivial top paths, hence $T'_1$ and $T'_2$ share at least $w-t'+t$ border edges, and by Lemma~\ref{lem:main}, the flip sequence from $T_2'$ to $T_1'$ costs at most $n-w+t'-t$. The total length of the flip sequence between $T_1$ and $T_2$ is thus at most $2n-2w+t'-t=2n-t+1$.
\end{proof}

Lemma~\ref{lem:paraedge} immediately implies:

\begin{corollary}\label{coro:longpath}
Let $T_1,T_2$ be two non-crossing spanning trees such that $T_1$ contains a subpath of length $t$. There exists a flip sequence between $T_1$ and $T_2$ of length at most $2 n-\frac t3$.
\end{corollary}
\begin{proof}
Let $Q:=x_1,\ldots,x_{t+1}$ be a subpath of $T_1$ of length $t$.
For every $2 \le i \le t-1$, we say that the edge $x_ix_{i+1}$ of $Q$ is \emph{separating} if $x_{i-1}$ and $x_{i+2}$ are separated by $x_i,x_{i+1}$ (\emph{i.e.} exactly one of $x_i,x_{i+1}$ appear between $x_{i-1}$ and $x_{i+2}$ in the cyclic ordering of the vertices). We say that the edge is a \emph{series edge} otherwise. By convention, the first and last edges of $Q$ are both series and separating. Denote by $s$ (resp. $p$) the number of series edges (resp. separating edges), so that $s+p=t+2$.

Observe that the set of separating edges of $Q$ are parallel, hence Lemma~\ref{lem:paraedge} ensures that there exists a flip sequence of length at most $a= 2n-\frac{p-1}2$. Moreover, if $x_ix_{i+1}$ is a series edges then there is a border edge in $T_1$ between $x_i$ and $x_{i+1}$ (in the part that does not contain the vertices $x_{i-1}$ and $x_{i+2}$). So there exists also a flip sequence of length at most $b=2n-s+1$ from $T_1$ to $T_2$ (passing through a border tree). Now observe that $2a+b = 6n-t$, hence either $a$ or $b$ must be at most $\frac{6n-t}{3}$, which concludes.
\end{proof}

In the case of paths, we can actually improve  Corollary~\ref{coro:longpath} by finding a flip sequence of length  at most $\frac 32 n$ (reproving a result of~\cite{Aichholzer22+} in a shorter way):

\begin{theorem}\label{thm:onepath}
Let $T_1,T_2$ be two non-crossing spanning trees such that $T_1$ is a path. There exists a flip sequence between $T_1$ and $T_2$ of length at most $\frac 32 n$. 
\end{theorem}
\begin{proof}
Let $x_1,\ldots,x_n$ be the vertices of the path $T_1$ (in order). Deleting $x_1$ and $x_n$ in the cyclic ordering yields two sets of vertices, one called the top part and the other the bottom part. We consider that $x_1$ and $x_n$ appear in both parts. Observe that all the edges of $T_1$ are either border edges (between two consecutive vertices of the top or the bottom part) or \emph{traversing edges} with one endpoint in each part.

Let us denote by $n_t$ (resp. $n_b$) the number of vertices of the top part (resp. bottom part) including $x_1$ and $x_n$. Note that $n_t+n_b=n+2$.
Let us denote by $b_t,b_b$ the number of border edges in $T_1$ respectively in the top and bottom parts.

Now we add all the $n_t+1$ border edges of the top part to $T_2$ and transform in $T_1$ all the $b_b$ border edges of the bottom part into traversing edges (similarly to Claim~\ref{clm:simple}). Observe that the two resulting trees share $b_t$ common border edges in the top part, and satisfy the hypothesis of Lemma~\ref{lem:main}. Therefore there is a flip sequence of length at most $n-b_t$ between them, and thus we can transform $T_1$ in $T_2$ with at most $(n_t-1)+b_b+(n-b_t)$ flips.

Exchanging the top and bottom parts in the previous argument yields another flip sequence of length $(n_b-1)+b_t+(n-b_b)$. The sum of these lengths is at most $2n+n_b+n_t-2 = 3 n$ which ensures one of the two sequences has length at most $\frac 32 n$, which completes the proof.
\end{proof}


\begin{thebibliography}{1}

\bibitem{Aichholzer22+}
Oswin Aichholzer, Brad Ballinger, Therese Biedl, Mirela Damian, Erik~D.
  Demaine, Matias Korman, Anna Lubiw, Jayson Lynch, Josef Tkadlec, and Yushi
  Uno.
\newblock Reconfiguration of non-crossing spanning trees.
\newblock {\em CoRR}, abs/2206.03879, 2022.
\newblock \href {http://arxiv.org/abs/2206.03879} {\path{arXiv:2206.03879}},
  \href {https://doi.org/10.48550/arXiv.2206.03879}
  {\path{doi:10.48550/arXiv.2206.03879}}.

\bibitem{AvisF96}
David Avis and Komei Fukuda.
\newblock Reverse search for enumeration.
\newblock {\em Discret. Appl. Math.}, 65(1-3):21--46, 1996.
\newblock \href {https://doi.org/10.1016/0166-218X(95)00026-N}
  {\path{doi:10.1016/0166-218X(95)00026-N}}.

\bibitem{HernandoHMMN99}
M.~Carmen Hernando, Ferran Hurtado, Alberto M{\'{a}}rquez, Merc{\`{e}} Mora,
  and Marc Noy.
\newblock Geometric tree graphs of points in convex position.
\newblock {\em Discret. Appl. Math.}, 93(1):51--66, 1999.
\newblock URL: \url{https://doi.org/10.1016/S0166-218X(99)00006-2}.

\end{thebibliography}
\end{document}